\documentclass{article}

\usepackage{arxiv}

\usepackage[utf8]{inputenc} 
\usepackage[T1]{fontenc}    
\usepackage{hyperref}       
\usepackage{url}            
\usepackage{booktabs}       
\usepackage{amsfonts}       
\usepackage{nicefrac}       
\usepackage{microtype}      
\usepackage{lipsum}
\usepackage{graphicx}
\usepackage{mathtools}
\usepackage{bigints}
\usepackage{amssymb, amsmath, amsthm}

\graphicspath{ {./images/} }

\newtheorem{theorem}{Theorem}[section]

\newtheorem{lemma}[theorem]{Lemma}

\title{Using Weighted P-Values in Fisher's Method}

\date{June 16, 2020}

\author{
 Arvind Thiagarajan \\
  \texttt{arvindthiagarajan@gmail.com}
  }

\begin{document}
\maketitle
\begin{abstract}
Fisher's method prescribes a way to combine p-values from multiple experiments into a single p-value. However, the original method can only determine a combined p-value analytically if all constituent p-values are weighted equally. Here we present, with proof, a method to combine p-values with arbitrary weights.
\end{abstract}

\section{Introduction}

For a chosen null hypothesis, there may be many ways to compute a p-value given a measurement, but all these ways share a common structure. A summary statistic is chosen, the cumulative distribution function (CDF) for the summary statistic is computed, and finally the value of the CDF at the observed value of the summary statistic is reported as the p-value. Because this process must always pass through a CDF in this manner, p-values can be reasoned about without knowing exactly which summary statistics were used in their computation.

Given a series of experiments designed to test the same null hypothesis, \href{https://en.wikipedia.org/wiki/Fisher's_method}{Fisher's method} provides a way to determine a p-value for the combined set of experiments, given the p-values for the individual experiments. However, Fisher's method is not equipped to incorporate prior differential confidences that a reasoner may have in the individual experiments. We will begin by describing the original Fisher's method before deriving a corresponding method that utilizes weights.

\section{Background}

We will describe the original Fisher's method here, beginning with proofs for two relevant, well-known results.

\begin{lemma}
Conditioned on the null hypothesis being true, the p-value of any observation will be uniformly distributed between $0$ and $1$.
\end{lemma}

\begin{proof}
Let $p(x)$ be the probability density function (PDF) and $P(x)$ be the cumulative density function (CDF) for observation $x$ under the null-hypothesis, and let $q(y)$ be the probability density function for the p-value obtained in such an experiment. Equating the parameterizations (under $x$ and $y$) of the probability mass over an infinitesimal region of space, we have

\begin{equation}
    \label{eq:1}
    q(y) dy = p(x) dx \implies \frac{dy}{dx} = \frac{p(x)}{q(y)}
\end{equation}

However, we also have that $y = P(x)$ by definition, which gives

\begin{equation}
    \label{eq:2}
    \frac{dy}{dx} = \frac{dP}{dx} = p(x)
\end{equation}

\noindent Equating \eqref{eq:1} and \eqref{eq:2} gives $q(y) = 1$, as desired.
\end{proof}

\begin{lemma}
The negative log of a random variable drawn from the uniform distribution between 0 and 1 will follow an exponential distribution.
\end{lemma}

\begin{proof}
Let the original (uniformly distributed) variable be $x$ and the new variable and distribution be $y = -\log{x}$ and $p(y)$ respectively. The probability mass of an infinitesimal region under the $x$-parameterization is $1 * |dx|$, while the probability mass under the $y$-parameterization is

\begin{equation}
    \label{eq:3}
    p(y) |dy| = p(y) |\frac{dy}{dx} dx| = p(y) \frac{|dx|}{x}
\end{equation}

Equating both expressions and simplifying gives $p(y) = x = e^{-y}$, as desired.

\end{proof}

Let \{$x_i$\} be the p-values obtained from a set of experiments testing the same null hypothesis. The original Fisher's method \cite{reference} uses 

\begin{equation}
    \label{eq:4}
    S = -\sum_{i}{\log{x_i}}
\end{equation}

as a summary statistic, noting that under the null hypothesis, $2S$ should follow a Chi-Squared distribution. However, if we were to introduce confidence weights $w_i$ corresponding to our $x_i$, then

\begin{equation}
    \label{eq:5}
    S(\{w_i\}) = -\sum_{i}{w_i \log{x_i}}
\end{equation}

would still be a valid summary statistic, but $2S$ would no longer follow a Chi-Squared distribution. In the final section of this paper, we derive a method, with proof, for computing the CDF of (and thus the combined p-value associated with) $S(\{w_i\})$

\section{Weighted Fisher's Method}

Let $\{v_i, \forall i \in \mathbb{Z} | 1 \leq i \leq k\}$ be a set of exponentially distributed random variables and $\{w_i, \forall i \in \mathbb{Z} | 1 \leq i \leq k\}$ be a set of fixed positive weights. In order to use Fisher's method with weights, we need to characterize the PDF and CDF of $V = \sum_{i}{w_i v_i}$. In the next two subsections, we will consider two kinds of weight sets - those in which all weights are distinct, and those in which all weights are identical. In the third subsection, we will describe how to handle an arbitrary weight set using the results from these two cases. We do note that the latter case (identical weights) is already handled by the standard implementation of Fisher's method, but we include it here for completeness.

Before we present our results and associated proofs, it's worth first stating that the distribution of $v = w_i v_i$ for any $i$ is 

\begin{equation}
    \label{eq:6}
    \displaystyle p(v) = \frac{1}{w_i} e^{-\frac{v}{w_i}}
\end{equation}

Also, while we have seen neither the proofs we present here nor the full generality of our results in prior work, we have observed specific cases \cite{alves}\cite{good} of our derived results in use throughout the broader statistics community. 

\subsection{Distinct Weights}

Before we begin the main proof, we would like to establish a smaller result.

\begin{lemma}
\begin{equation}
    \label{eq:7}
    \sum_{i=1}^{k+1}{\frac{{w_i}^{k-1}}{\prod_{j \neq i}{(w_i - w_j)}}} = 0
\end{equation}
\end{lemma}

\begin{proof}
Let $\displaystyle h(z) = \frac{z^{k-1}}{\prod_{j=1}^{k+1}{(z-w_j)}}$. Now consider the path $C$ running counterclockwise along the circle of radius $R$ centered at the origin in the complex plane. For arbitrarily large $R$, $C$ contains all the residues of $h$, which are in fact just the various $w_i$. As such, the residue theorem gives us that 

\begin{equation}
    \label{eq:8}
    \sum_{i=1}^{k+1}{\frac{{w_i}^{k-1}}{\prod_{j \neq i}{(w_i - w_j)}}} = \frac{1}{2\pi i}\int_{C}{h(z)dz}
\end{equation}

\noindent The magnitude of the integral on the right hand side of \eqref{eq:8} can be bounded above by

\begin{equation}
    \label{eq:9}
    \int_{C}{\frac{\|z\|^{k-1}}{\prod_{j}{\|z-w_j\|}}\|dz\|} \leq  \frac{2 \pi R^k}{(R - \max_i(w_i))^{k+1}}
\end{equation}

\noindent Taking $R \rightarrow \infty$, we have that

\begin{equation}
    \label{eq:10}
    |\sum_{i=1}^{k+1}{\frac{{w_i}^k}{\prod_{j \neq i}{(w_i - w_j)}}}| = \frac{1}{2\pi} \|\int_{C}{h(z)dz}\| \leq  \frac{R^k}{(R - \max_i(w_i))^{k+1}} \rightarrow (R - \max_i(w_i))^{-1} \rightarrow 0
\end{equation}

\noindent proving the desired result.
\end{proof}

\begin{theorem}
The PDF for the weighted sum $V$ of $k$ exponentially distributed variables when all weights are distinct is

\begin{equation}
    \label{eq:11}
    p_k(V) = \sum_{i=1}^{k}{\frac{{w_i}^{k-2} e^{-\frac{V}{w_i}}}{\prod_{j \neq i}{(w_i - w_j)}}}
\end{equation}
\end{theorem}

\begin{proof}
We will prove this result using induction. For the base case of $k=1$, it is easily confirmed that \eqref{eq:11} matches \eqref{eq:6}. Now for the inductive hypothesis, suppose we add weight $w_{k+1}$ and variable $v_{k+1}$ to our sets and redefine $V$ to include $w_{k+1}v_{k+1}$ in the sum. It follows that

\begin{equation}
    \label{eq:12}
    p_{k+1}(V) = \int_{0}^{V}{p_k(z) \frac{1}{w_{k+1}} e^{-\frac{V-z}{w_{k+1}}} dz}
\end{equation}

\noindent To evaluate this integral, we first perform an intermediate computation for $i\leq k$:

\begin{equation}
    \label{eq:13}
    f(i) = \int_{0}^{V}{\frac{1}{w_i}e^{-\frac{z}{w_i}}\frac{1}{w_{k+1}} e^{-\frac{V-z}{w_{k+1}}} dz} = \frac{1}{w_i-w_{k+1}}\left(e^{-\frac{V}{w_i}} - e^{-\frac{V}{w_{k+1}}}\right)
\end{equation}

\noindent Substituting \eqref{eq:11} into \eqref{eq:12}, exchanging the summation with the integral, and simplifying with \eqref{eq:13} gives

\begin{equation}
    \label{eq:14}
    p_{k+1}(V) = \sum_{i=1}^{k}{\frac{{w_i}^{k-1} f(i)}{\prod_{(j \leq k) \neq i}{(w_i - w_j)}}} = \left(\sum_{i=1}^{k}{\frac{{w_i}^{k-1} \left(e^{-\frac{V}{w_i}} - e^{-\frac{V}{w_{k+1}}}\right)}{\prod_{j \neq i}{(w_i - w_j)}}}\right) 
\end{equation}

\noindent \eqref{eq:14} can be separated into two sums as

\begin{equation}
   \label{eq:15}
   p_{k+1}(V) = \left(\sum_{i=1}^{k}{\frac{{w_i}^{k-1} e^{-\frac{V}{w_i}}}{\prod_{j \neq i}{(w_i - w_j)}}}\right) - e^{-\frac{V}{w_{k+1}}}\left(\sum_{i=1}^{k}{\frac{{w_i}^{k-1} }{\prod_{j \neq i}{(w_i - w_j)}}}\right) 
\end{equation}

\noindent and multiplying \eqref{eq:7} by $\displaystyle e^{-\frac{V}{w_{k+1}}}$ and adding the result to \eqref{eq:15} gives

\begin{equation}
   \label{eq:16}
   p_{k+1}(V) = \sum_{i=1}^{k+1}{\frac{{w_i}^{k-1} e^{-\frac{V}{w_i}}}{\prod_{j \neq i}{(w_i - w_j)}}}
\end{equation}
\noindent thus proving the inductive hypothesis and the desired result.
\end{proof}

\begin{theorem}
The right-tail p-value for the weighted sum $V$ of $k$ exponentially distributed variables when all weights are distinct is

\begin{equation}
   \label{eq:17}
    \sum_{i=1}^{k}{\frac{{w_i}^{k-1} e^{-\frac{V}{w_i}}}{\prod_{j \neq i}{(w_i - w_j)}}}
\end{equation}
\end{theorem}
\begin{proof}
\begin{equation}
   \label{eq:18}
    \int_{V}^{\infty}{p_k(z) dz} = \sum_{i=1}^{k}{\frac{{w_i}^{k-2} \int_{V}^{\infty}{e^{-\frac{z}{w_i}}dz}}{\prod_{j \neq i}{(w_i - w_j)}}} = \sum_{i=1}^{k}{\frac{{w_i}^{k-1} e^{-\frac{V}{w_i}}}{\prod_{j \neq i}{(w_i - w_j)}}}
\end{equation}
\end{proof}

\subsection{Identical Weights}

\begin{theorem}
The PDF for the weighted sum $V$ of $k$ exponentially distributed variables when all weights are identical ($w_i = w$ $\forall$ $i$) is

\begin{equation}
   \label{eq:19}
    p_k(V) = \frac{1}{(k-1)!w}\left(\frac{V}{w}\right)^{k-1}e^{-\frac{V}{w}}
\end{equation}
\end{theorem}

\begin{proof}
We will prove this result using induction. For the base case of $k=1$, it is easily confirmed that \eqref{eq:19} matches \eqref{eq:6}. Now for the inductive hypothesis, we have that $p_{k+1}(V)$ is

\begin{equation}
   \label{eq:20}
    \int_{0}^{V}{p_k(z) \frac{1}{w} e^{-\frac{V-z}{w}} dz} = \frac{w^{-k-1}e^{-\frac{V}{w}}}{(k-1)!}\int_{0}^{V}{z^{k-1} dz} = \frac{e^{-\frac{V}{w}}}{k!w}\left(\frac{V}{w}\right)^{k}
\end{equation}

\noindent thus proving the inductive hypothesis.
\end{proof}

\noindent Before we extend this to find an expression for the p-value we're interested in, we will prove an intermediate result using differentiation under the integral sign.

\begin{lemma}
\begin{equation}
    \label{eq:21}
    \int_{c}^{\infty}{u^n e^{-au} du} = c^n \left(\sum_{j=0}^{n}{j!\binom{n}{j}(ca)^{-j}}\right)\frac{e^{-ca}}{a}
\end{equation}
\end{lemma}

\begin{proof}
\begin{equation}
    \label{eq:22}
    \int_{c}^{\infty}{u^n e^{-au} du} = \int_{c}^{\infty}{(-1)^n \left(\frac{d}{da}\right)^n \left(e^{-au}\right) du}
\end{equation}

\noindent Exchanging the order of differentiation and integration in (20) gives

\begin{equation}
    \label{eq:23}
    \left(\frac{d}{da}\right)^n\int_{c}^{\infty}{(-1)^n e^{-au} du} = (-1)^n \left(\frac{d}{da}\right)^n\left(a^{-1}e^{-ca}\right)
\end{equation}

\noindent Using the binomial expansion of the product rule, we have

\begin{equation}
    \label{eq:24}
    \left(\frac{d}{da}\right)^n\left(a^{-1}e^{-ca}\right) = \sum_{j=0}^{n}{\binom{n}{j}\left(\frac{d}{da}\right)^j(a^{-1})\left(\frac{d}{da}\right)^{n-j}(e^{-ca})}
\end{equation}

\noindent Simplifying \eqref{eq:24}, we have that this expression becomes

\begin{equation}
    \label{eq:25}
    \sum_{j=0}^{n}{\binom{n}{j}(-1)^j j! (a^{-j-1}) (-c)^{n-j} e^{-ca}} = (-c)^n\left(\sum_{j=0}^{n}{j!\binom{n}{j} (ca)^{-j}}\right) \frac{e^{-ca}}{a}
\end{equation}

\noindent We substitute \eqref{eq:25} back into \eqref{eq:23} and simplify to obtain the desired result.
\end{proof}

\begin{theorem}
The right-tail p-value for the weighted sum $V$ of $k$ exponentially distributed variables when all weights are identical is

\begin{equation}
    \label{eq:26}
    \left(\sum_{i=0}^{k-1}{\frac{1}{i!}\left(\frac{V}{w}\right)^i}\right)e^{-\frac{V}{w}}
\end{equation}
\end{theorem}
\begin{proof}
Let $u = \frac{z}{w}$. Then we have

\begin{equation}
    \label{eq:27}
    \int_{V}^{\infty}{p_k(z) dz} = \frac{1}{(k-1)!}\int_{\frac{V}{w}}^{\infty}{u^{k-1} e^{-u} du}
\end{equation}

\noindent Substituting \eqref{eq:21}, with $c = \frac{V}{w}, a=1,$ and $n=k-1$, into \eqref{eq:27}, reversing the order of the summation, and reducing the binomials and factorials gives the desired result. 
\end{proof}

\subsection{Handling Arbitrary Sets of Weights}

In most situations, we may not expect our weights to be uniformly identical or pair-wise distinct. In these cases, we cannot directly apply our derived expressions to the full set of weighted p-values. However, we propose here a computational method that will allow for handling arbitrary sets of weights.

Let $V$ again denote the weighted sum of the negative log p-values. Additionally, let $\{w_i\}$ denote the set of \textit{unique} weights and $\{n_i\}$ denote the multiplicities of the weights, i.e. $w_i$ appears $n_i$ times. Finally, for every $i$, let $a_i = (w_i)^{-1}$ and let $V_i$ denote the weighted sum of all $n_i$ negative log p-values associated with $w_i$. Having proven theorem 3.4, we know that the PDF $p_i$ for $V_i$ follows the form given by \eqref{eq:19}, with $k=n_i$ and $w=w_i$.

Let $p(V)$ denote the PDF for $V$. Since $V = \sum_{i}{V_i}$, it follows that $p(V)$ is the convolution of all the $p_i$. Now, let $f(s)$ be the Laplace Transform (LT) of $p(V)$ and let $f_i$ be the LT of $p_i$ for each $i$. Using \eqref{eq:19}, along with the fact that the LT converts convolutions to products, we have that

\begin{equation}
    \label{eq:28}
    f(s) = \prod_{i}{f_i(s)} = \prod_{i}{\frac{{a_i}^{n_i}}{(s+a_i)^{n_i}}}
\end{equation}

Furthermore, we have that the LT, $F$, of the left-tail p-value for $V$ is given by

\begin{equation}
    \label{eq:29}
    F(s) = \frac{f(s)}{s} = \frac{1}{s}\prod_{i}{\frac{{a_i}^{n_i}}{(s+a_i)^{n_i}}}
\end{equation}

As can be seen in \eqref{eq:29}, $F(s)$ is a rational fraction, and so can be decomposed into the following form:

\begin{equation}
    \label{eq:30}
    F(s) = \frac{c}{s} + \sum_{i}{\sum_{j=1}^{n_i}{\frac{c_{ij}}{(s+a_i)^{j}}}}
\end{equation}

and consequently the reverse LT can be applied to \eqref{eq:30} to obtain

\begin{equation}
    \label{eq:31}
    P(V) = c + \sum_{i}{\sum_{j=1}^{n_i}{\frac{c_{ij}}{(j-1)!} V^{j-1} e^{-a_i V}}}
\end{equation}

 Determining the values of $c$ and the $c_{ij}$ that will render \eqref{eq:29} and \eqref{eq:30} equal is a special case of partial fraction decomposition, a problem that has been very well studied \cite{brugia} \cite{pfe}. Recursive expressions for these coefficients exist that can be computed quite efficiently \cite{pfe}. 
 
 In fact, these recursions can be solved explicitly to obtain analytic expressions for $c$ and the $c_{ij}$, but we leave such expansions as exercises for the interested reader. The derivation of such expansions should not require tools or techniques beyond those used in this paper. Additionally, the LT-based methodology we have outlined in this subsection for deriving an expression for $P(V)$ can be used, even without the more sophisticated results on partial fraction decomposition, to provide alternative proofs for theorems 3.2, 3.3, 3.4, and 3.6.
 
  As we have stated, the values of $c$ and the $c_{ij}$, and through them the function $P$, are accessible to us for an arbitrary set of weights. By construction, $P(V)$ denotes the left-tail p-value of $V$, and so $1-P(V)$ is the right-tail p-value of $V$, i.e. the combined p-value that we set out to compute.
 
 An implementation of this procedure for computing the combined p-value for an arbitrary set of weights can be found on GitHub, under arvindthiagarajan/multimodal-statistics.

\bibliographystyle{plain}  
\bibliography{references}

\nocite{*}

\end{document}